\documentclass[12pt,a4paper]{article}
\usepackage{amssymb}
\usepackage{amsmath,amsfonts,amsthm}
\usepackage{amstext,latexsym,bm}
\usepackage{graphicx}
\usepackage{amsmath}
\usepackage{xcolor}

\usepackage{amsfonts}
\usepackage{amsthm,amscd}
\usepackage{graphicx}
\usepackage{amsmath}
\usepackage{amssymb}
\usepackage{amstext}

\newtheorem{theorem}{Theorem}

\newtheorem{remark}[theorem]{Remark}

\title{On the quantum Guerra-Morato Action Functional}

\author{J. Knorst and A. O. Lopes}

\begin{document}
\maketitle

\abstract{Given a smooth potential $W:\mathrm{T}^{n} \to \mathbb{R}$ on the torus,
the Quantum Guerra-Morato action functional  is given by
\smallskip

$ \,\,\,\,\,\,\,\,\,\,\,\,\,\,\,\,\,\,\, \,\,\,\,\,\,\,\,\, I(\psi) = \int\,(\, \, \,\frac{D v\, D v^*}{2}(x) - W(x) \,) \,\,a(x)^2 dx,$
\smallskip

\noindent
where $\psi $ is described by $\psi = a\, e^{i\,\frac{ u }{h}} $, 
$ u =\, \frac{v + v^*}{2},$ $a=e^{\,\frac{v^*\,-\,v}{2\, \hbar} }$, $v,v ^*$ are real functions, $\int a^2 (x) d x =1$, and $D$ is derivative on $x \in \mathrm{T}^{n}$.  It is natural to consider the constraint $ 
\mathrm{d}\mathrm{i}\mathrm{v}(a^{2}Du)=0$, which means flux zero. The $a$ and $u$ obtained from a
 critical   solution (under variations $\tau$)  for such action functional, fulfilling  such constraints,  satisfy the Hamilton-Jacobi equation with a quantum potential. Denote $'=\frac{d}{d\tau}$. We show that the expression for the second variation of a critical solution is given by 

\smallskip
$\,\,\,\,\,\,\,\,\,\,\,\,\,\,\,\,\,\,\, \,\,\,\,\,\,\,\,\,\,\,\,\,\,\,\,\,\,\,\,\,\,\,\,\,\,\,\int a^{2}\,D[ v' ]\, D [(v ^*)']\, dx.$

\smallskip

Introducing the constraint   $\int a^2 \,D u \,dx =V$, we also consider later an associated dual eigenvalue problem. From this follows a transport and a kind of eikonal equation.
}
\section{Introduction}\label{Sec2.21}

In \cite{Ev} L. C.  Evans considered an Action Functional which is related to Aubry-Mather theory
and also present another  Action Functional which he called the Guerra-Morato Action Functional. As we will see this last one is more adequate for the description of the Physics of Quantum Mechanics, and to analyze its properties is the primary goal of the present work. In \cite{Ev} a few properties of this last mentioned functional were presented and we  intend to provide here a more complete description of the topic. One of our goals is to characterize critical action solutions for   this quantum action functional; they are not necessarily local minimal solutions. Our main result concerns the expression for the second variation of a critical solution (see Theorem \ref{poir1}). We will present full proofs adapting the reasoning followed in \cite{Ev} (with the necessary changes at each moment).

The Evans Functional Action (see (1.5) in \cite{Ev}) plays an important role in the quantum analog of weak KAM theory (a different topic not covered here).

The reason for the terminology Guerra-Morato Action Functional introduced in \cite{Ev} was 
motivated by the paper \cite{GM} by F. Guerra and L. Morato (see also  expression (14.3) in \cite{Nelson}, \cite{Voll} and \cite{Yu}). In fact, the expression coined by Evans, which is of a stationary nature,  does not appear in this form in \cite{GM} which is of nonstationary type (see the Lagrangian field expression (93) in \cite{GM}). Moreover, in \cite{GM} there is no assumption to constraints which is an essential issue in \cite{Ev} and here. 

The papers
\cite{Ana1}, 
\cite{Ana2}, 
\cite{Ana3}, \cite{LT},  \cite{GLM}, \cite{Zanelli}, \cite{Zanelli1}, \cite{G1} and \cite{Ber}   analyze different types of problems related to  Evans Functional Action and Mather measures.

  Let's consider a potential of class $C^\infty$ given by $W$ : $\mathrm{T}^{n}=(S^1)^n\rightarrow \mathbb{R}$, where
  $S^1$ is the unit circle. Under some suitable assumptions similar results hold for $W: \mathbb{R}^n \to \mathbb{R}$ but we will not address this issue here.

The classical  Lagrangian is given by
$$
L(v,\ x)=\frac{m}{2}|v|^{2}-W(x),\,\,\,x \in \mathrm{T}^{n},\, v\in \mathbb{R}^{n},
$$
and the associated classical  (mechanical) Hamiltonian is
$$
H(p,\ x)=\frac{1}{2\, m}|p|^{2}+W(x),\,\,\,x \in \mathrm{T}^{n} ,\,p\in \mathbb{R}^{n},
$$
where $m$ denotes the mass.
\noindent

Taking $p=m\, v$, it is well known  that the solutions of the Euler-Lagrange equation for $L$ and the ones for the Hamilton equation correspond to each other via a change of coordinates. It is well known that in Classical Mechanics there exist critical action principles for the Lagrangian  formulation and  for Hamiltonian one (see \cite{Arnold}). From the point of view of Physics, solutions minimizing the action are preferred in the theory.

We denote $ \hbar= \frac{1}{m}.$ A wave function has the form $\psi = a\, e^{i\,\frac{ u }{h}} $, where $a $ is positive and $u$ is real. $D$ will denote the derivative with respect to $x$.

\medskip

In Quantum Mechanics the expected value of the Hamiltonian is
\begin{equation}
\int (\frac{\hbar}{2}| D\, \psi|^2 + W |\psi|^2 )\, d x
\end{equation}

Given the functions $a>0$ and $u$, where $a:T^n \to \mathbb{R}$ and $u:T^n \to \mathbb{R}$,
  consider $s=\hbar\,\log a$ and take $v^* =  ( u+s)$
and $v = (u-s)$, then
  $$  \frac{v+v^*}{2} =u$$
  and
  $$ \frac{v^*- v}{2} =s =\hbar\,\log a.$$

  Thus, we can write every function of the form $a\, e^{i\,\frac{u}{\,\hbar} } $,
   in terms of $v$ and $v^*$ by  solving  the above equation.

   If $u$ and $a$ are periodic (\,in $(S^1)^n\,)$ the same goes for $v,v^*:T^n \to \mathbb{R}$.

   Likewise given $v$ and $v^*$, the functions $a>0$ and $u$ can be determined using the above expressions.

The Guerra-Morato action functional (see (7.10) in \cite{Ev}) is given by
\begin{equation} \label{pri} I(\psi) = \int\,(\, \, \,\frac{D v\, D v^*}{2} - W \,) \,\,a(x)^2 dx
\end{equation}
where $\psi $ is described by $\psi = a\, e^{i\,\frac{ u }{h}} $, 
$ u =\, \frac{v + v^*}{2},$ $a=e^{\,\frac{v^*\,-\,v}{2\, \hbar} }$, $v,v ^*$ are real functions and $\int a^2 (x) d x =1$. In some parts of our work we will take $\hbar=1$ to simplify the notation.

We will consider variations $\tau$ of certain given $\psi$ and we denote $'=\frac{d}{d\tau}$

When considering critical action for general variations and for such functional, we are analyzing a problem that is different from the  critical action problem  for the Evans Action Functional, as defined by  (1.5) in \cite{Ev} (see also  our Remark \ref{klj}). The expression \eqref{pri} is of Lagrangian type and the main point is that the critical solutions will determine solutions of the Hamilton-Jacobi equation where the potential has the right sign (this does not happen when considering the Evans Action Funcional; as in (2.8) in \cite{Ev}). This is the reason why the use of \eqref{pri} is more suitable from the physical point of view.

The function $\psi$ will be assumed to satisfy the constraints  \eqref{fra2}, \eqref{fra3}, \eqref{fra4}
to be defined next (see also \eqref{par4}, \eqref{par5} and \eqref{par6}). Expressions  \eqref{fra2}, \eqref{fra3}, \eqref{fra4} are in some sense the {\it field version} of the corresponding classical constraints in Aubry-Mather  theory (see \cite{Fathi}). The expression \eqref{fra3} corresponds to the assumption that a probability is holonomic  and condition  \eqref{fra4} plays the role of a homological class (see Section 2.6 in \cite{CI}).

Our main result is

\begin{theorem}  \label{poir1} If $\psi=a\,e^{i\,u/h}$ is critical for the action and $a>0$, where $ u =\, \frac{v + v^*}{2},$ $a=e^{\,\frac{v^*\,-\,v}{2\, \hbar} }$,
then the second derivative of the variation has the expression
\begin{equation} \label{rwr1}
\,\, \int a^{2}\,D[ v' ]\, D [(v ^*)']\, dx.
\end{equation}
\end{theorem}
%The functions $v$ and $v^*$ have momentum status. 

Here we will  analyze a principle of {\bf critical action} for such an action $I$. The sign of \eqref{rwr1} will determine if the critical solution is a local minimum or a local  maximum (or indeterminate).

We say that a certain $\psi$ is critical for $I$ if any close-by variation will have zero derivative. Another issue will be whether at a given critical $\psi$,  the functional $I$ attains a local minimum, a local maximum, or neither. Our objective will be to analyze the first and second variational problems.

We point out that  in Sections \ref{FiVar} and \ref{SeVar} here (in the same way as in Section 2 in \cite{Ev}) the functions $v,v^*:T^n \to \mathbb{R}$, do not satisfy any other constraint different from $\int   e^{\,\frac{v^*(x)\,-\,v(x)}{\, \hbar} }d x=1.$ When comparing our setting  with the hypothesis 
(16), (17), and (18) in \cite{Voll}, it is clear that in \cite{Voll} there are assumptions on $b,b^*$, which does not correspond to the freedom of our $v,v^*.$ In the same direction, expressions  (4) (5) (6), and (7)  in  \cite{Carlen}  show that we are considering a different class of problems.

 In Section \ref{outra} and \ref{duaal} the functions 
$v,v^*$ will satisfy certain relations which will be described by  \eqref{elr1} and \eqref{elr2}. When considering \eqref{rwr1} (see also \eqref{elr1} and \eqref{elr2}) we get from Remarks \ref{yyt3} and \ref{yyt4} that a critical point is not  necessarily a local minimum (or maximum).

In our understanding, the  above expression of the functional  \eqref{pri}, taken from Section 7.2 in \cite{Ev}, corresponds to a different type of problem  that one considered in \cite{GM}. We also point out that \cite{GM}
considered   variational problems depending on time which is a different framework when compared with \cite{Ev}.

%Related results appear in \cite{Ana1}, \cite{Ana2}, \cite{Ana3}, \cite{Carlen}, \cite{Kupka}, \cite{G1}, \cite{GISY}, \cite{Zanelli}, \cite{Yu}, \cite{Nelson}.

\begin{remark} \label{kpo}

To add the term $P x$ to $u$, getting $\tilde{u}= u + P\, x$,  is equivalent to consider the new functions  $\tilde{v} = v + P\, x$ and  $\tilde{v}^* = v^* + P\, x$.
\end{remark}

%This functional $I$ is not positive definite.

%The results of this section are inspired (and adapted) from \cite{Ev} that considers another action functional (of Lagrangian nature and not Hamiltonian as above) that we can call Evans action functional.

Below $D$ will denote derivative with respect to $x$.
\medskip

We will need a later expression \eqref{trri}.

As $ \frac{a'}{a}= \,\frac{\,(v^*)' -v'\, }{\, 2\,\hbar} $, then 

\begin{equation} \label{trri}
-\,\hbar^2\, \Big|D\Big(\frac{a^{\prime}}{a}\Big)\Big|^{2}\,\,+\, \, \,|Du^{\prime}|^ {2} \,=\, \,D[ v' ]\, D [(v^*)']\,\,
\end{equation}

In the case we want to consider $\tilde{u}=u + P x$ instead of $u$, we get

\begin{equation} \label{trri1}
-\,\hbar^2\, \Big|D\Big(\frac{a^{\prime}}{a}\Big)\Big|^{2}\,\,+\, \, \,|D\tilde{u}^{\prime}|^ {2} \,=\, \,D[\tilde{ v}' ]\, D [(\tilde{v}^*)']\,\,=\, \,D[ v' ]\, D [(v^*)'].
\end{equation}

%\section{Quantum 8}

\smallskip

\noindent

   Here $\tilde{\psi}$ will be considered in the form $\tilde{\psi}=a\, e^{i \frac{u}{\hbar}},$ where $a=e^{\,\frac{\,v^* -v\, }{2\,\hbar }}$ and $ u\,= \frac{ v^* +v}{2}$, where $v $ and $v^*$ are periodic functions taking real and differentiable values.

Let's assume that $\psi$ has the "Bloch waveform"
$$
\psi=e^{i\,\,  \frac{1}{\hbar} P\ x}\, \tilde{\psi}
$$
for some $P\in \mathbb{R}^{n}$ and where $\tilde{\psi}$ : $\mathrm{T}^{n}\rightarrow \mathbb{C}$ was given in the above form. The introduction of the parameter $P$ is quite important in our reasoning.

Consider a fixed vector $V\in \mathbb{R}^{n}$.

We consider here  the Quantum Guerra-Morato action of the state $\psi $
\begin{equation} \label{fra1}
A[\psi]\text{ }:=\int_{\mathrm{T}^{n}}(\,\,\frac{D \,v^*(x) \, \,D\, v ( x)}{2}-\,W(x))\, a^{2}(x)dx,
\end{equation}
where we assume that
\begin{equation} \label{fra2}
\int_{\mathrm{T}^{n}}|\psi(x)|^{2}dx=\int_{\mathrm{T}^{n}}a(x)^{2}dx=1 ,\text{ }
\end{equation}
\begin{equation} \label{fra3}
\int_{\mathrm{T}^{n}}(\overline{\psi(x)}D\psi(x)-\psi(x) D\overline{\psi(x)})\cdot D\ \phi(x) dx=0\text{ for all}\text{ }\phi\in C^{1}(\mathrm{T}^{n}),\text{ }
\end{equation}
and
\begin{equation} \label{fra4}
\frac{\hbar}{2 i\,}\int_{\mathrm{T}^{n}}(\overline{\psi}D\psi-\psi D\overline{\psi})\, dx=V.\text{ }
\end{equation}
\bigskip
\bigskip

%Here $\hbar$ is positive and called the Plank constant.

\noindent
\quad If $\psi$ is differentiable then \eqref{fra3} implies

\noindent

\begin{equation} \label{fra5} \mathrm{d}\mathrm{i}\mathrm{v}(\overline{\psi}D\psi-\psi D \overline{\psi})=0
\end{equation}

The vector field
$j :=\overline{\psi}D\psi-\psi D\overline{\psi}$ represents what is called the flux in quantum mechanics.

\noindent

Note that
\begin{equation} \label{hi} \frac{\hbar}{ 2 \, i} \,(\, \overline{\psi}D\psi-\psi D\overline{\psi}\,)= \frac{\hbar}{ 2 \, i} \,\Big(a a' + a^2 \frac{i \, u'}{\hbar}
- a \, a' + a^2 \frac{i \, u'}{\hbar}\Big)= a^2 \,u'. \end{equation}

\noindent

So the condition \eqref{fra5} can be replaced by the transport equation
$$ \text{div} (a^2 \, Du)=0$$
and condition \eqref{fra4} by
\begin{equation} \label{hip1}  \int a^2 \,D u \,dx =V.
\end{equation}

Fixing $v,v^*$ we will denote
$$V_0 :=   \int a^2(x) \,D u(x) \,dx = \int \big( e^{\,\frac{\,v^*(x) -v(x)\, }{2\,\hbar }}\big)^2\, D \Big( \frac{ v^*(x) +v(x)}{2} \Big) d x.$$

If we consider $\tilde{u}=\frac{ v^*(x) +v(x)}{2}+ P x$ instead of $u=\frac{ v^*(x) +v(x)}{2}$, as in Remark \ref{kpo}, we get
\begin{equation} \label{hip2}  \int a^2 \,D \tilde{u} \,dx = \int a^2 \,(D u + P) \,dx = V_0 + P=V= V_P.
\end{equation}

\begin{remark} \label{iutw}
This shows that the value $V$, which we fix  as a constraint, is coupled with the term $P$. Note that in the reasoning of Sections \ref{FiVar} and \ref{SeVar} the constraint  \eqref{hip1} is not used.
\end{remark}

\section{First variation} \label{FiVar}

\noindent

\noindent
\quad Consider the state $\psi$ in polar form where
\begin{equation} \label{par1}
\psi=ae^{i\, \frac{u}{\hbar}},
\end{equation}
and where the phase $u$ satisfies the expression
\begin{equation} \label{par2}
u(x)=\,P\cdot x+z(x)
\end{equation}
for some periodic function $z$.

It is known from (7.11) Section 7.2 in \cite{Ev} that the Guerra-Morato action is given in this case by
\begin{equation} \label{par3}
A[\psi]=\int_{\mathrm{T}^{n}}-\,\frac{\hbar^2 }{2}\, |Da|^{2}+\frac{\, a^{2}}{2\, }|Du|^{2}-\, Wa^{2}dx,
\end{equation}
\medskip

\begin{remark} \label{klj}
This action is different from (2.3) considered in \cite{Ev} due to a change of sign in the first term under the integral (see section 7.2 in \cite{Ev}). Indeed, compare \eqref{par3} with (2.3) in \cite{Ev}.
\end{remark}

We are initially interested in the critical points of this action.

\noindent

\noindent

\noindent

Note that  (\ref{fra2}), (\ref{fra3}) and (\ref{fra4}) become by (\ref{hi}) in the expressions
\begin{equation} \label{par4}
\int_{\mathrm{T}^{n}}a^{2}dx=1,\text{ }
\end{equation}
\begin{equation} \label{par5}
\mathrm{d}\mathrm{i}\mathrm{v}(a^{2}Du)=0,\text{ }
\end{equation}
\begin{equation} \label{par6}
\,\int_{\mathrm{T}^{n}}a^{2}\,\,Du\,\,dx=V.
\end{equation}

\medskip

Let $\{(u(\tau),\ a(\tau))\}_{-1\leq\tau\leq 1}$ be a differentiable family of functions indexed by $\tau$ such that it satisfies the constraints above, i.e., (\ref{par4})-(\ref{par6}) holds for $a=a(\tau)$, and moreover $(u(\mathrm{0}),\ a(\mathrm{0}))=(u,\ a)$. Let's assume that for every $\tau\in(-1,1)$, we can write
$$
u(\tau)=P(\tau)\cdot x+z(\tau),
$$
where $P(\tau)\in \mathbb{R}^{n}$ and $v(\tau)$ are $\mathrm{T}^{n}$-periodic. Define
$$
j(\tau)\text{ }:=\int_{\mathrm{T}^{n}}-\frac{\,h^{2}}{2m}|Da(\tau)|^{2} +\frac{\, a^{2}(\tau)}{2}|Du(\tau)|^{2}-\, \, W \,a^{2}(\tau)dx,\text{ }
$$

\medskip
 In the next theorems, we denote $'=\frac{d}{d\tau}$, and  in turn, the derivative with respect to $x$ will be denoted as $D$.
\medskip.

We can thus consider variations and, under the hypothesis that the choice of $a,u$ is critical for the action,
we will be able to obtain constraints for $a$ and $u$.

Consider $a$ and $u$ fixed and  we will look for conditions obtained  when they are critical for the action.

Note that as  the functions $a'=a'(0)$ and $u'=u'(0)$ denote the derivative with respect to the variation $\{(u(\tau),\ a(\tau)) \}_{-1\leq\tau\leq 1}$, with the variable $\tau$, they can be quite general. But, for example,
$a^{\prime}$ must satisfy the identity
$$\int_{\mathrm{T}^{n}} 2\, \frac{ d\,a (\tau)}{d \, \tau}\,a\,\,dx=\,2\,
\int_{\mathrm{T}^{n}}a^{\prime}\,a\,\,dx=0,
$$
which is obtained by differentiating with respect to $\tau$, at $\tau=0$, the expression
$$
\int_{\mathrm{T}^{n}}a^{2}\,\,dx=1.
$$

\noindent

\begin{theorem} Suppose that $\psi $ is described by $\psi = a\, e^{i\,\frac{u}{\,\hbar} } $.
Then $j^{\prime}(\mathrm{0})=0$ for all variations, if and only if,
\begin{equation} \label{this}
\frac{\, h^{2}}{2}\triangle a=a\left(\,\frac{|Du|^{2}}{2}\, +\,W-E\right)\text{ }
\end{equation}
for some real number $E$.

In this case, $a$ and $u$ satisfy the so-called Hamilton-Jacobi equation with a quantum potential.

\end{theorem}

\begin{proof}
First, let's estimate $j'(\tau)$. An easy account shows that
$$
j^{\prime}(\tau)=\int_{\mathrm{T}^{n}} \left(\,-\frac{h^{2}}{m}\, Da\cdot Da ^{\prime}+aa^{\prime}|Du|^{2}+a^{2}Du\cdot Du^{\prime}\,-\,2\, Waa^{\prime}\,\right)dx.
$$

Above, $a=a(\tau), \, u=u(\tau)$. Taking derivatives at $\tau=0$ under the assumption of constraints we obtain:
\begin{equation} \label{nana1}
\mathrm{d}\mathrm{i}\mathrm{v}(2aa^{\prime}Du\ +a^{2}Du^{\prime})=0,
\end{equation}
\begin{equation} \label{nana2}
\int_{\mathrm{T}^{n}}\,(\mathit{2}aa^{\prime}Du\text{ }+a^{2}Du^{\prime})\, dx=0 ,\text{ }
\end{equation}
Remember that $Du=P+D z$. We multiply (\ref{nana1}) by $u$ and integrate in space. Then, by integrating by parts we arrive at 
$$
\int_{\mathrm{T}^{n}}(2aa^{\prime}|Du|^{2}+a^{2}Du^{\prime}\cdot Du)\,d x=0.
$$
So,
\begin{align*}
j^{\prime}(0) &=\int_{\mathrm{T}^{n}}\left(-\, \frac{h^{2}}{m}Da\, Da'-a\,a'|Du|^{2}-2\, Wa\,a'\right) dx
\\
&=2\int_{\mathrm{T}^{n}}a^{\prime}\left(\frac{h^{2}}{2m}\triangle a- \left(\frac{|Du|^{2}}{2}+\, W\right)a\right)\,dx.
\end{align*}
\medskip
\noindent

\noindent

Remember that each variation $a^{\prime}$ must satisfy the identity
$$\int_{\mathrm{T}^{n}} \frac{ d\,a (\tau)}{d \, \tau}\,a\,\,dx=
\int_{\mathrm{T}^{n}}a^{\prime}\,a\,\,dx=0.
$$
%Note also that $a$ and $u$ depend on $m$.

Assume that $a$ and $u$ satisfy
$$\,\frac{h^{2}}{2m}\triangle a-\left(\frac{|Du|^{2}}{2 }+\, W\right)a\,= -\,E\, a$$
for some constant $E$.

Thus, $j^{\prime}(0)=0$ for any variation of $a$ and $u$.

The above expression means
\begin{equation} \label{retior1}
\left(\frac{|Du|^{2}}{2}+\, W\right)\,-\,E\,=\, \frac{ h^{2}}{2m}\frac{\triangle a}{a},
\end{equation}
for some constant $E$.

\medskip

\medskip
Let's show that when $j'(0)=0$ we get that  the reciprocal holds. We assume that $a^2>0$. 

Now,
\begin{align*}
j^{\prime}(0) &= 2 \int_{\mathrm{T}^{n}} a^\prime \,\left(\frac{h^{2} }{2m}
\triangle a\,-a \left(\frac{|Du|^{2}}{2}+\, W\right)\,\right) dx \\
& = 2 \int_{\mathrm{T}^{n}} a^\prime \,g \, dx.
\end{align*}

As $a^\prime$ is general, the only restriction being that $\int aa^\prime \, dx = 0$, we can conclude from $j^\prime(0)=0$, that there is a constant $E$ such that $g=-a\,E$ and the result follows (similar reasoning as Theorem 2.1 in \cite{Ev}).

Note that given the constraint \eqref{par6}  on $V$, the vector $P$ has to be set, and this shows a coupling of the value $E$ with $V$ (and also $P$)

\medskip

\end{proof}

\bigskip

The next result was proved in   Theorem 7.2 in \cite{Ev}

\begin{theorem} \label{ptre}
If $\psi = a \, e^{i\,\,\frac{u}{\hbar}} $ is differentiable and critical for the Guerra-Morato action, then $\psi$ is an eigenfunction of the
Hamiltonian  operator
\begin{equation} \label{retio} -\, \frac{\hbar^2}{2} \Delta \psi + W \psi =E\, \psi,  
\end{equation}
for some $E\in \mathbb{R}.$

\end{theorem}
\medskip

Note that expression \eqref{this} is different from expression (2.8) in \cite{Ev} due to a change of sign in the Laplacian term.

  \bigskip

\section{Second variation} \label{SeVar}

\noindent
Let us now analyze the second variation. Thus, we will take the second derivative  of $j$  with respect to $\tau$.

\noindent
\begin{theorem} Suppose that $\psi $ is described by $\psi = a\, e^{i\,\frac{u}{\,\hbar} } $, where
$ u =\frac{v + v^*}{2},$ $a=e^{\frac{\,v^* -v\, }{2\, \hbar} }$ and $v,v ^*$ are real functions.
If $\psi=ae^{i\, u/h}$ is a critical point for action $I$ then
\begin{equation} \label{ert1}
j^{\prime \prime}(0)=\int_{\mathrm{T}^{n}}-\frac{h^{2}}{m}|Da^{\prime}|^{2 }+a^{2}|Du^{\prime}|^{2}-
2(a^{\prime})^{2}\left(\frac{|Du|^{2}}{2}+ W\,-E\,\right)dx.
\end{equation}
\end{theorem}

\begin{proof} Note that
\begin{align}\label{ert2}
j^{\prime\prime}(\tau) &=\int_{\mathrm{T}^{n}} \,\,\left[-\frac{h^{2}}{m}|Da^ {\prime}|^{2}-\frac{h^{2}}{m} Da\cdot Da^{\prime\prime}\,\right] \nonumber \\ 
& \qquad \quad + [aa^{\prime\prime}|Du|^{2}+4aa^{\prime}Du\cdot Du^{\prime}+a^{2}Du\cdot Du^{\prime\prime}\,] \nonumber \\
& \qquad \qquad  + [\,(a^{\prime})^{2}|Du|^{2}+ a^{2}|Du^{\prime}|^{ 2} -2\, W(a^{\prime})^{2}\, -2\, Waa^{ \prime\prime}\,]\, \,dx.
\end{align}

Taking derivative in (\ref{nana1}) and (\ref{nana2}) we get
\begin{equation} \label{ert3}
\mathrm{d}\mathrm{i}\mathrm{v}(2(a^{\prime})^{2}Du+\mathit{2}a a^{\prime \prime}Du +4aa^{\prime}Du^{\prime}+a^{2}Du^{\prime \prime})=0\end{equation}
and
\begin{equation} \label{ert4}
\int_{\mathrm{T}^{n}}(\,2(a^{\prime})^{2}Du+\mathit{2}aa^{\prime\prime}Du\text{ }+4 aa^{\prime}Du^{\prime}\text{ }+a^{2}Du^{\prime\prime}\,)\,dx=0.\end{equation}
\quad

% and
% \begin{equation} \label{ert4}
% \int_{\mathrm{T}^{n}}(\,2(a^{\prime})^{2}Du+\mathit{2}aa^{\prime\prime}Du\text{ }+4 aa^{\prime}Du^{\prime}\text{ }+a^{2}Du^{\prime\prime}\,)\,dx=0.\end{equation}
\quad

Now take $\tau=0$, multiply (\ref{ert3}) by $u$, integrate, and do integration by parts to get
\begin{equation} \label{ert5}
\int_{\mathrm{T}^{n}}\left[\,2(a^{\prime})^{2}\,|Du|^{2}+2aa^{\prime\prime}|Du|^{ 2}+\mathit{4}aa^{\prime}Du^{\prime}\cdot Du+a^{2}Du^{\prime\prime}\cdot Du\,\right]dx=0 .\end{equation}
\smallskip

Let's now use the expression (\ref{ert2}). As $j'(0)=0$, we integrate by parts  $(-h^{2}/2m) \,  Da\cdot Da^{\prime\prime}$ and obtain
\begin{align*}
j^{\prime\prime}(0) &= \int_{\mathrm{T}^{n}} \Big[ 2a^{\prime\prime}\left(\frac{h^{2}}{2m}\triangle a-a\left( \frac{|Du|^{2}}{2}+ \, \, W\right)\right)-\frac{h^{2}}{m}|Da^{\prime}|^{2} \\
& \qquad-(a^{\prime})^{2}|Du|^{2}+a^{2}|Du^{\prime}| ^{2}-2\,\, W(a^{\prime})^{2}\, \,\Big] \, dx
 \\
&=\int_{\mathrm{T}^{n}}2a^{\prime\prime}(-E\,\,a)-\frac{h^{2}}{m}|Da^{\prime}|^{2}-2(a^{\prime})^{2}\left(\frac{|Du|^{2}}{2}+ \, W\right)+a^{2}|Du^{\prime}|^{2}\, \,  dx \\
&=\int_{\mathrm{T}^{n}}-\frac{h^{2}}{m}|Da^{\prime}|^{2}+a^{2}|Du^{\prime}|^{2}-2(a^{\prime})^{2}\left(\frac{|Du|^{2}}{2}+ \, W- \,E\right) dx. \\
& =\int_{\mathrm{T}^{n}}-\frac{h^{2}}{m}|Da^{\prime}|^{2}+a^{2}|Du^{\prime}|^{2}  \,\, dx-
\int_{\mathrm{T}^{n}}\,2(a^{\prime})^{2}\left(\frac{\,|Du|^{2}}{2}+ \, W- \,E\right)dx.
\end{align*}
\medskip
\noindent

Above we use  the identity 
$$
\int_{\mathrm{T}^{n}}a^{\prime\prime}a+(a^{\prime})^{2}\,\, dx=0,
$$
obtained by differentiating with respect to $\tau$ twice and using $\int a^2 \,dx=1$.

The reasoning above proves the claim we were looking for.

\end{proof}

\medskip

Note that \eqref{ert1} is an expression which is different from (2.12) in \cite{Ev}.
\medskip

\section{Other expression for $j''(0)$} \label{outra}

We want to get a more appropriate expression for $j''(0)$.

\medskip

Let's assume that
$$
a(x)>0,\text{ for all }\,x \in \mathrm{T}^{n}.
$$
\noindent
\begin{theorem}  \label{poir} If $\psi=a\,e^{i\,u/h}$ is critical for the action and $a>0$, where $ u =\, \frac{v + v^*}{2},$ $a=e^{\,\frac{v^*\,-\,v}{2\, \hbar} }$,
then
$$ j^{\prime\prime}(0)=\int_{\mathrm{T}^{n}}\left(|Du^{\prime}|^{2} - \frac{h^2}{m} \Big|D\Big(\frac{a^{\prime}}{a}\Big)\Big|^{2}\right) \,a^{2}\, dx = $$
\begin{equation} \label{rwr}
\int a^{2}\,D[ v' ]\, D [(v ^*)'] dx.
\end{equation}
\end{theorem}

\smallskip
\noindent
{\bf Proof:} Since $\psi$ is critical, form \eqref{retior1} we have
$$
\frac{\,h^{2}}{2m}\triangle a=a\left(\frac{\,|Du|^{2}}{2}+W-E\right).
$$

So, from (\ref{trri}) we have
\begin{align*}
j^{\prime\prime}(0) &=\int \, a^{2}|Du^{\prime}|^{2}-\frac{h ^{2}}{m}\,  \,|Da^{\prime}|^{2}-\,\,2(a^{\prime}) ^{2}\left(\frac{\,h^{2}}{2m}\frac{\triangle a}{a}\right)\,dx \\
& =\int \,a^{2}|Du^{\prime}|^{2}-\frac{h^{2}}{m}\, |Da^{\prime}|^{2}-
\,\frac{h^{2}}{m} (a^{\prime})^{2}\, \frac{|Da|^{2}}{a^{2}}+\,\frac{2h^{2}}{m}a^{\prime}\, \frac{Da^{\prime}\cdot Da}{a} \,\, dx \\
&=\int_{\mathrm{T}^{n}}\,a^{2}|Du^{\prime}|^{2}-\frac{h^{2}}{m}\, a^ {2}\Big|D\Big(\frac{a^{\prime}}{a}\Big)\Big|^{2}\,dx=\, \int a^{2}\,D [ v' ]\, D [(v ^*)'] \, dx,
\end{align*}
showing the claim of the  theorem.

$\square$

\medskip

Note that 
$j^{\prime\prime}(0)$ in expression \eqref{rwr} may be positive or negative.

\begin{remark} \label{yyt3}
Remember that $\psi $ is described by $\psi = a\, e^{i\,\frac{u}{\,\hbar} } $, where
$ u =\frac{v + v^*}{2},$ $a=e^{\frac{\,v^*- v}{2\, \hbar} }$ e $v,v ^*$ are real functions. Note that if $v=- v^*$ then the $j^{\prime\prime}(0)$ is negative. Also, if $v+v^*$ is constant, then the above expression is
negative. 

\end{remark}

\begin{remark} \label{yyt}
Note that in the case we consider $u= \frac{v + v^*}{2} +  P \, x$,  we  also get
 \begin{equation}  \label{wer} j^{\prime\prime}(0) = \int a^{2}\,D [ v' ]\, D [(v ^*)'] \, dx.
\end{equation}

We point out that from Remark \ref{iutw} the vector $P$ has to be set from the constraint \eqref{par6}  on $V$.

\end{remark}

\smallskip

\begin{remark} \label{yyt4}

Similar results hold when $x \in \mathbb{R}^n$ and not on the torus. Let's consider $x\in \mathbb{R}$ next.
Note that in the case of the harmonic oscillator
$H(x,p)= \frac{p^2}{2 m} + \frac{m\, w^2 \, x^2 }{2} $ we have the ground state
$\psi_0$, which is the minimum energy state $E_0= \frac{1}{2} \, \hbar \, w$, is described by
$$\psi_0(x)= ( \frac{m \, w}{\pi \, \hbar} )^{\frac{1}{4}} e^{- \frac{m\, w\, x^2}{2\, \hbar}}.$$
Thus, $|\psi_0|^2$ will determine a Gaussian density with variance $a= \sqrt{\frac{h}{2\,m\, w}}.$

Note that $ v^{*} (x)-v(x)=-m\, x^2$ and $v +v^{*} =0$. 

Thus, $v ' =- v^{* } $ $' $ and

$$\, \frac{1}{m}\int a^{2}\,D[ v' ]\, D [(v^*)']\,\,dx<0$$

\noindent
  That is, we just show  the ground state $\psi_0$ is not a local minimum  for the Guerra-Morato action $I$ when $x\in \mathbb{R}$ (and not in the circle). 
  
  \end{remark} 
  
  %But we will show that under the conditions of the next section (in the circle), we will always get $\, \frac{1}{m}\int a^{2}\,D[ v' ]\, D [(v^*)']\,\,dx<0$.
  
  %{\bf The  $P$ and $E^0$, are fixed from the constraint $V$. Consider all $\psi(x)=a(x)\,e^{i\, (u(x)+P x)/h}$ satisfying the $V$ constraint, associated to $E^0= \frac{1}{2} \, \hbar \, w$}

%\end{document}
\medskip

\noindent

\section{Dual eigenfunctions} \label{duaal}

\noindent

In this section, we will assume conditions on $v,v^*$.

\medskip

\noindent
It is natural to consider the dual eigenvalue problems:
\begin{equation} \label{elr1}
\left\{\begin{array}{ll}
-\frac{\,h^{2}}{2}\triangle w
+\, h P\cdot Dw +\,W\,w=E^{0}_{\hbar}w & \mathrm{i}\mathrm{n}\ \mathrm{T}^{n}\\
w\ \mathrm{i}\mathrm{s}\ \mathrm{T}^{n} \, \text{periodic}
\end{array}\right.
\end{equation}

and

\begin{equation} \label{elr2}
\left\{\begin{array}{ll}
-\frac{\,h^{2}}{2}\triangle w^*- \,h P.\ Dw^* +\,Ww^* =E^{0}_{\hbar} w^* & \text{in}\, \mathrm{T}^{n}\\
w^*\ \mathrm{i}\mathrm{s}\ \mathrm{T}^{n}\, \text{periodic},
\end{array}\right.
\end{equation}

\noindent where $E^{0}_{\hbar}(P)\in \mathbb{R}$ is the main eigenvalue.  We may assume the real eigenfunctions $w,\ w^{*}$ to be positive in $\mathrm{T}^{n}$ and normalized in such away  that
\begin{equation} \label{qual}
\int_{\mathrm{T}^{n}}w w^{*} dx=1.
\end{equation}

Moreover, we can take $w,\ w^{*}$ and $E^{0}_{\hbar}$ to be twice differentiable  in $x$ and $P$.

Note the change of sign in the $W$ term when compared with (3.1) and (3.2) in \cite{Ev} (and also \cite{Ana1} and \cite{Ana3}).

The equations \eqref{elr1} and \eqref{elr2} correspond - in the classical point of view - to the equations satisfied by, respectively, the left eigenfunction  and the right eigenprobability for the action of a continuous time semigroup (not stochastic) as described in \cite{BEL}, \cite{KLMN}, \cite{LMN} (for a related quantum result see \cite{BKL}). Equation \eqref{qual} for  $w\, w^*$ corresponds to the equilibrium probability for  the associated renormalized (Gibbs) stochastic semigroup. In some sense, these are results analogous to the ones obtained in the discrete-time setting through the use of  the Ruelle operator theorem. The time reversible  setting is related to  the case where $w=w^*$ (as in \cite{KLMN}) .

\noindent
\quad In the same way as in \cite{Ev} we employ a form of the Cole-Hopf transformation, to define
$$
\left\{\begin{array}{l}
v:=-h\log w\\
v^{*}:=h\log w^{*}.
\end{array}\right.
$$
Then,
$$
\left\{\begin{array}{l}
w=e^{-v/h}\\
w^{*}=e^{v^{*}/h}
\end{array}\right.
$$
and then follows that

\begin{equation} \label{LuBo1}
\left\{\begin{array}{ll}
- \frac{\,h}{2\,}\triangle v+\frac{1}{2}|P+Dv|^{2}- W=\overline{H}_{\hbar}(P) & \mathrm{i}\mathrm{n}\ \mathrm{T}^{n}\\
v\ \mathrm{i}\mathrm{s}\ \text{ a}\,\mathrm{T}^{n}-\text{periodic function}   
\end{array}\right.
\end{equation}

and 
\begin{equation} \label{LuBo2}
\left\{\begin{array}{ll}
\frac{\,h}{2\,}\triangle v^*+\frac{1}{2}|P+Dv^*|^{2}-W=\overline{H}_{\hbar}(P) & \mathrm{i}\mathrm{n}\   \text{is a}\,  \mathrm{T}^{n}\\
v^*\ \mathrm{i}\mathrm{s}\  \text{ a}\,  \mathrm{T}^{n}-\text{periodic function}
\end{array}\right.
\end{equation}

for
\begin{equation} \label{LuBo7}
\overline{H}_{\hbar}(P)\text{ }:=\frac{\,|P|^{2}}{2}-E^{0}_{\hbar}(P).
\end{equation}

It follows from classical  PDE estimates the bounds
$$
|Dv|,\text{ }|Dv^{*}|\leq C,
$$
for a constant $C$ depending only on $P$ and  $W$.

%Note that (3.7) is exactly the equation we get on (2.8) if $u(x)=Px + v^*$.

\medskip

\begin{remark}  The equation (\ref{LuBo1}) in the one-dimensional case admits a solution 
\[ v = -h \, \log (w) +Cx \]

\noindent where $w>0$ solves the equation 

\[ w'' - 2\left(\frac{P+C}{h}\right) w' + \left(\left(\frac{P+C}{h}\right)^2-\frac{2}{h^2}(\bar{H}_h+W)\right) w=0.   \]

For $v^*$ take 
\[ v^* = h \, \log (u) +Cx \]

\noindent where $w^*>0$ solves the equation 

\[ -(w^*)'' - 2\left(\frac{P+C}{h}\right) (w^*)' + \left(\left(\frac{P+C}{h}\right)^2-\frac{2}{h^2}(\bar{H}_h+W)\right) w^*=0.   \]

\bigskip

Now we define
 \begin{equation}  \label{oculos1}
a^2 =\sigma\text{ }:=ww^{*}=e^{\frac{v^*-v}{h}}\,\,\,\,\,
\end{equation}
and
 \begin{equation}  \label{oculos2}
u\text{ }:=P\cdot x+\frac{v+v^{*}}{2}.
\end{equation}

Note that although $w,\ w^{*},\ v,\ v^{*},\ u$ and $\sigma$ depend on $h$, we will for notational simplicity mostly not write these functions with a subscript $h$. The importance of the product \eqref{oculos1} of the eigenfunctions is also noted in \cite{Ana3} (see also \cite{LT}) but  it is used for a different purpose related to Aubry-Mather Theory (see \cite{Fathi}).

Note that
$$a\, e^{\frac{i\, u}{h}}= e^{\frac{\frac{1}{2}(v^{*}-v ) \,+\,i \,\frac{v+v^{*}}{2}+\, i\, P\cdot x }{h} } .$$
\noindent
\quad According to \eqref{qual},
$$
\sigma>0\text{ }in\text{ }\mathrm{T}^{n},\text{ }\int_{\mathrm{T}^{n}}\sigma dx=1.
$$
\end{remark}

\smallskip
\begin{theorem} \label{BC3} For
$$
u=P\cdot x+\frac{v+v^{*}}{2},\,\,\text{we get}
$$

 (i)
\begin{equation} \label{BC1}
\mathrm{d}\mathrm{i}\mathrm{v}\,(\sigma\, Du)\,=0\,\,\text{  in}\,\,\mathrm{T}^{n}.
\end{equation}
(ii) {\it Furthermore},

\noindent
\begin{equation} \label{BC2}  \frac{1}{2}|Du|^{2}-\,W-\overline{H}_{\hbar}(P)=\frac{\,h}{4}\triangle(v- v^*)-\frac{1}{8}|Dv-Dv^{*}|^{2} \,\,  \text{ in}\,\, \mathrm{T}^{n}. 
\end{equation}

\end{theorem}

\noindent
\quad We call   \eqref{BC1} the {\it continuity} (or {\it transport}) {\it equation}, and regard \eqref{BC2} as an {\it eikonal equation} with an error term on the right-hand side. Note that the form of \eqref{BC2} is not exactly the classical Hamilton-Jacobi equation due to the minus sign multiplying the potential $W$ (in the left-hand side of \eqref{BC2}). This expression is different from (3.12) in \cite{Ev}.

\noindent

\begin{proof}
a) Note that
\begin{align*}
h\,\mathrm{d}\mathrm{i}\mathrm{v}(w^{*}Dw-wDw^{*})&=h\,(w^{*}\triangle w-w\triangle w^{ *}) \\
&=\frac{2}{h}\big(w^{*}\big(\frac{h^{2}}{2}\triangle w\big)-w\big(\frac{h^{2}}{2}\triangle w ^{*}\big)\big)
\\ &=\frac{2}{h}[w^{*}(-E^{0}w+W w+h P\cdot Dw)
\\ & \qquad+w(E^{0}w^{*}-Ww^{*}+hP\cdot Dw^{*})] \\
&=2 (w^{*}P\cdot Dw+w\,P\cdot Dw^{*})=2P\cdot D\sigma.
\end{align*}

But
$$
w^{*}Dw\text{ }-wDw^{*}=w^{*}\Big(-\frac{Dv}{h}w\Big)-w\Big(\frac{Dv^{*}}{h}w ^{*}\Big)=-\frac{1}{h}\sigma(Dv+Dv^{*}),
$$
and therefore
$$
P\cdot D\sigma+\frac{1}{2}\mathrm{d}\mathrm{i}\mathrm{v}(\sigma \mathrm{D}(\mathrm{v}+\mathrm{v}^ {*}))=0.
$$
This shows  (i).

%\end{document}

\noindent
b)  Taking into account the expression (see \cite{Ev})
$$
\frac{1}{2}|a-b|^{2}+\frac{1}{2}|a+b|^{2}=|a|^{2}+|b|^{2},
$$
and taking $a=(P+Dv)$, $b=( P+Dv^{*})$,  we get

\begin{align*}
\frac{1}{2}\left|P+\frac{1}{2}D(v+v^{*})\right|^{2}&=\frac{1}{8}|(P+Dv)+( P+Dv^{*})|^{2}
\\ &=\frac{1}{4}|P+Dv|^{2}+\frac{1}{4}|P+Dv^{*}|^{2}-\frac{1}{8}| Dv-Dv^{*}|^{2}.
\end{align*}

Therefore,  we finally get

$$
\frac{1}{2}|Du|^{2}-W-\overline{H}_{\hbar}(P) =\frac{1}{2}\left(\frac{ 1}{2}|P+Dv|^{2}-W -\overline{H}_{\hbar}(P)\right) $$
$$+\frac{1}{2}\left(\frac{1}{2}|P+Dv^{*}|^{2}-W-\overline{H}_{\hbar}(P)\right)-\frac {1}{8}|D(v-v^{*}|^{2} 
=$$
$$\frac{1}{2}\big(\frac{h\,}{2}\triangle v\big)-\frac{1}{2}\big(\frac{h}{2}\triangle v^{*}\big)-\frac{1}{8}|D(v-v^{*})|^{2}=$$
$$\frac{\hbar}{4}\triangle (v- v^{*})-\frac{1}{8}|D(v-v^{*})|^{2}.$$
\end{proof}

\noindent
\begin{remark} One can also show that
$$
-\frac{h}{2}\triangle\sigma-\mathrm{d}\mathrm{i}\mathrm{v}((P+Dv)\sigma)=0,
$$
$$
-\frac{h}{2}\triangle\sigma+\mathrm{d}\mathrm{i}\mathrm{v}((P+Dv^{*})\sigma)=0.
$$

Indeed, note that
$$
\frac{\,h}{2}\triangle\sigma=\,\mathrm{d}\mathrm{i}\mathrm{v}\left(\frac{1}{2}D(v^{*}-v )\sigma\right).
$$
Now, add and subtract the above from (\ref{BC1}).

\end{remark}

\noindent

\medskip

Now we will show integral identities involving $Du$ {\it and} $D^{2}u$. To simplify notation, we will denote
$$
 d\sigma\text{ }:=\sigma dx.
$$

\begin{theorem}

\begin{equation} \label{BCB1}
\int_{\mathrm{T}^{n}}\frac{1}{2}|Du|^{2}-W d\sigma=\overline{H}_{\hbar}(P)+\frac{1\,}{8}\int_{\mathrm{T}^{n}}|Dv-Dv^{*}|^{2}d\sigma.
\end{equation} 

 \end{theorem}
 
\medskip
\noindent

\noindent
\begin{proof}  Note that from the above
$$
\int_{\mathrm{T}^{n}}\frac{1}{2}|Du|^{2}-W-\overline{H}_{\hbar}(P)  d\sigma=$$
$$\frac{\,h}{4}\int_{\mathrm{T}^{n}}\hspace{-5px} \triangle(v- v^*)\sigma dx -\frac{1}{8}\int_{\mathrm{T}^{n}}|Dv-Dv^{*}|^{2}\sigma dx.
$$
But $\sigma=ww^{*}=e^{\frac{v^{*}-v}{h}}$, and therefore
\begin{align*}
\int_{\mathrm{T}^{n}}\frac{1}{2m}|Du|^{2}-W-\overline{H}_{\hbar}(P)d\sigma &=-\frac{\,h}{4}\int_{\mathrm{T}^{n}}D(v-v^{*})\cdot\frac{D(v^{*}-v)}{h}\, \sigma dx
\\
& \qquad -\frac{1}{8}\int_{\mathrm{T}^{n}}|Dv-Dv^{*}|^{2}\, \sigma dx \\ 
&=\frac{1}{8}\int_{\mathrm{T}^{n}}|Dv-Dv^{*}|^{2}d\sigma.
\end{align*}
\noindent
\end{proof}

\noindent
\bigskip
\noindent

\begin{theorem}

For $\hbar $ fixed,
taking derivative with respect to $P$
\begin{equation} \label{polir0} i)\,\,\,\frac{d}{ dP} H_\hbar(P)= \int Du\, d \sigma.
\end{equation}
Moreover,
\begin{align}\label{polir}
 ii)\,\,\,
 \frac{d^2}{ dP^2}\overline{H}_{h}(P) &=\int_{\mathrm{T}^{n}}D_{xP}^{2}u\otimes D_{xP}^{2}ud\sigma \nonumber \\
& \qquad + \frac{1}{4}\int_{\mathrm{T}^{n}}D_{xP}^{2}(v-v^{*})\otimes D_{xP}^{2}(v-v^{*})d\sigma.
\end{align}

This shows that  $\overline{H}_{h}$  is a convex function of $P$.

\end{theorem}
\begin{proof} The proof i) is similar to the proof of item 1) in Theorem 4.1 in \cite{Ev}; we just have to substitute $-W$ (of \cite{Ev}) by  $W$.
Taking derivative with respect to $P$ in \eqref{elr1} and \eqref{elr2} we get $\frac{d}{ dP} E^{0}_{\hbar}(P)=- \int D v \,d \sigma$ and 
$\frac{d}{ dP} E^{0}_{\hbar}(P)=- \int D v^* \, d \sigma$ (see page 321 in \cite{Ev}).

Then, from \eqref{LuBo7} we get
$$\frac{d}{ dP} H_\hbar(P) = P -\frac{d}{ dP}  E^{0}_{\hbar}(P)= P+ \frac{1}{2} \int  [Dv + Dv^*] d \sigma= \int  Du\, d \sigma. $$ 

\medskip

The proof  of ii) is also similar to the one in Theorem 4.1 in \cite{Ev}. 

The same cancellation procedure of item 2) in theorem 4.1 in \cite{Ev} results in \eqref{polir}.
\end{proof}

\medskip

The authors thank Diogo Gomes for helpful conversations on the topic.

\medskip

Data sharing is not applicable to this article as no new data were created or analyzed in this study.

\end{document}